\documentclass[a4paper,11pt]{article}
\usepackage{rawfonts}

\usepackage{url} 
\usepackage[ruled,lined]{algorithm2e}
\usepackage{graphicx}
\usepackage{tikz}
\usetikzlibrary{automata, positioning, arrows}

\usepackage{pgf}
\usepackage{pgfplots}
\usetikzlibrary{matrix,trees,arrows,arrows.meta,fit,shapes,positioning}
\usetikzlibrary{decorations.pathmorphing,decorations.pathreplacing,decorations.markings,decorations} 
\usepackage{paralist}

\usepackage{amssymb}
\usepackage{caption}
\usepackage{subcaption}

\usepackage{mymacros}

\pgfplotsset{compat=1.16}

\begin{document}

\title{Proving Regulatory Compliance: \\ Full Compliance Against an Expressive Unconditional Obligation is coNP-Complete}

\author{
  Silvano Colombo Tosatto$^1$\\
  \texttt{silvano.colombotosatto{@}data61.csiro.au}
  \and
  Guido Governatori\\
  \texttt{guido{@}governatori.net}
  \and
  Nick van Beest$^1$\\
  \texttt{nick.vanbeest@data61.csiro.au}
}
\date{%
    $^1$CSIRO Data61\\
    41 Boggo Rd, Dutton Park QLD 4102, Brisbane, Australia\\
    \vspace{1em}
    \today
}



\maketitle

\begin{abstract}
Organisations are required to show that their procedures and processes satisfy the relevant regulatory requirements. 
The computational complexity of proving regulatory compliance is known to be generally hard. 
However, for some of its simpler variants the computational complexity is still unknown. We focus on the eight variants of the problem that can be identified by the following binary properties: whether the requirements consists of one or multiple obligations, whether the obligations are conditional or always in force, and whether only propositional literals or formulae can be used to describe the obligations. This paper in particular shows that proving full compliance of a model against a single unconditional obligation whose elements can be described using formulae is co\textbf{NP}-complete.
Finally we show how this result allows to fully map the computational complexity of these variants for proving full and non compliance, while for partial compliance the complexity result of one of the variants is still missing.
\end{abstract}


\section{Introduction}\label{Ch_4.3_Sect:1}

The problem of proving compliance of process models consists of verifying whether the possible executions of these models satisfy some given requirements. Organisations are generally required to show that their procedures, which can be represented as process models, satisfy the regulatory requirements. The number of regulations and rules that the organisation needs to comply with is generally large \cite{morgan2002business}, and non-compliance may lead to significant fines and law suits \cite{zhang2007economic}.

Although the computational complexity of the general problem of proving regulatory compliance is known to be hard, the computational complexity of some of its variants is still unknown, and in this paper we provide the computational complexity for the following variant:

\begin{enumerate} 
\item Proving whether all executions of a process model satisfy a single requirement holding for the whole duration of the execution, and allowing the requirement to be described using propositional formulae.
\end{enumerate}

Considering the other variants of the problem that can be identified by the three following binary properties: \emph{a)} the number of obligations describing the regulatory requirements, \emph{b)} whether these requirements are conditional, and \emph{c)} whether these requirements conditions are restricted to be represented as single propositions, or can be represented using formulae. We position the new computational complexity result among the existing ones and show that we can now determine the computational complexity of all the variants when it comes to evaluating \emph{full} and \emph{non} compliance.\footnote{Whether all executions of a model, or none respectively, satisfy the given requirements.}

The remainder of the paper is structured as follows: Section~\ref{sec:preliminaries} introduces the formal concepts around business processes and business process models, together with the regulatory framework they are subject to. 
Section~\ref{sec:1g+} provides the new computational complexity result, which is placed in context
with existing results in Section~\ref{sec:overview}. 
Finally, Section~\ref{sec:summary} concludes the paper.


\section{Preliminaries}\label{sec:preliminaries}

\subsection{Acyclic Structured Process Models}

We first introduce process models, which represent a collection of possible ways in which an organisation can achieve a given objective. Verifying these models against some regulations corresponds to show that the ways the organisation achieves these objectives follow the regulations. 

We limit our study to problems dealing with \emph{acyclic structured} process models, similar to the structured workflows defined by Kiepuszewski et al. \cite{Kiepuszewski:2000:SWM:646088.679917}, which can be preliminarily defined as the models whose components are properly nested. 
Additionally, we consider only acyclic process models, meaning that these models do not contain loops and every task is limited to be executed at most once in each of the model's executions. 

Business processes can be modelled using different notations, which mostly use an intuitive graphical representation. 
Petri nets are a well-known tool for modelling concurrent processes, for which a rich body of theory and tools to verify their properties have been defined. 
In the context of this work, we use a specific sub-class called \emph{workflow nets} (WF nets), for which we define the formal framework below.

\subsubsection{Petri Nets}
A Petri net is a directed bipartite graph consisting of places, transitions, and arcs between transition-place pairs. 
The activities in a business process are represented as transitions in a Petri net, while the execution states are represented as combinations of places, which are token containers. A Petri net is formally defined as follows~\cite{Petri1966,Reisig1998}:

\begin{definition}[Petri net]\label{def:petrinet}
A \emph{Petri net} $\petri$ is a tuple $(\places,\transitions,\arcs, \labels)$, where:
\begin{compactitem}
  \item $\places$ is a finite set of \emph{places},
  \item $\transitions$ is a finite set of \emph{transitions}, such that $\places \cap \transitions =\emptyset$,
  \item $\arcs \subseteq (\places \times \transitions) \cup (\transitions \times \places)$ is a set of arcs, and
  \item $\labels: \places \cup \transitions \to \mathcal{L}$ is a labelling function. 
\end{compactitem}
\label{def:pn}
\end{definition}

Each place of a Petri net can contain a \emph{token}, and the \emph{net marking} (Definition \ref{def:netmarking}) of a Petri net is determined by which of its places contain tokens and which does not. A transition in a Petri net represents an operation towards the goal achieved by executing the Petri net. For instance, in a Petri net representing a cooking recipe, its transitions can correspond to the preparation steps required to prepare the recipe.

Depending on the net marking of a Petri net, some of its transitions can be executed, and their execution manipulate the tokens in the Petri net's places, altering the net marking. These correspond to the pre-set and post-set of a transition as described in Definition \ref{def:pre-post}.

\begin{definition}[Pre/Post-sets]\label{def:pre-post}
Given a Petri net $\petri = (\places,\transitions,\arcs, \labels)$, for each node $y \in \places \cup \transitions$:
\begin{itemize}
\item The \emph{pre-set} of $y$, denoted by $\pre{y}$, is the set of all input nodes of $y$: $\pre{y} = \{x \in \places \cup \transitions \ \mid \ (x,y) \in \arcs\}$.
\item The \emph{post-set} of $y$, denoted by $\post{y}$, is the set of all output nodes of $y$: $ \post{y}= \{x \in \places \cup \transitions \ \mid \ (y,x) \in \arcs\}$.
\end{itemize}
\end{definition}

\begin{definition}[Petri net state, net marking]\label{def:netmarking}
A \emph{net marking} $\markings$ of a Petri net $\petri = (\places,\transitions,\arcs, \labels)$ is a function that associates places with numbers, i.e., $\markings: \places \to \mathbb{N}_0$.
\end{definition}


\begin{definition}[Marked net]
A \emph{marked net} is a tuple $\petri= (\places,\transitions,\arcs, \labels,\markings_0)$, where $(\places,\transitions,\arcs, \labels)$ is a Petri net and $\markings_0$ is the initial marking.
\end{definition}

\begin{definition}[Enabled transition]
Given a net marking $\markings$, transition $t \in T$ is said to be \emph{enabled} if $\forall p \in \pre{t} : \markings(p) > 0$.
\end{definition}

%

\begin{definition}[Firing rule]
Given a marked net $\petri = (\places,\transitions,\arcs, \labels,\markings_0)$ and an enabled transition $t$, the \emph{firing} of $t$ is denoted by $\markings\larrow{t}\markings^\prime$ and leads to a new marking $\markings^\prime$ such that:
\begin{itemize}
\item $\forall p \in \pre{t}$, $\markings^\prime(p) = \markings(p) + 1$,
\item $\forall p \in \post{t}$, $\markings^\prime(p) = \markings(p) - 1$.
\end{itemize}
\end{definition}

Petri nets can be restricted with some extra properties, so that it resembles a business process having a single source place and a single sink place (i.e. a fixed start and end of the process). These special Petri nets are referred to as \emph{workflow nets (WF-nets for short)} and can be formally defined as follows:

\begin{definition}[WF-net]
A marked net $\petri = (\places,\transitions,\arcs,\labels,\markings_0)$ is a WF-net iff:
\begin{itemize}
  \item $\petri$ has a special source place $i$, such that $\pre{i} = \emptyset$ and $\forall p \in \places \setminus i$, $\pre{p} \neq \emptyset$. 
  \item $\petri$ has a special sink place $o$, such that $\post{o} = \emptyset$ and $\forall p \in \places \setminus o$, $\post{p} \neq \emptyset$. 
  \item if we add a transition $t_w$ to $\petri$ such that $\pre{t_w} = o$ and $\post{t_w} = i$, then $\petri$ is strongly connected.
\end{itemize}
\end{definition}

An execution of a WF-net is any sequence of transitions where the net marking resulting from a transition in the sequence is in the \emph{pre-set} of the next transition in the sequence.


\begin{definition}[Execution]\label{def:fullexe}
Given a WF-net $\petri = (\places,\transitions,\arcs,\labels,\markings_0)$, an execution $\exe$ of $\petri$ is a sequence of transitions $t_0, \dots, t_n$ such that:
\begin{itemize}
\item $\forall t_i \in \exe, t_i \in \transitions$,
\item $\pre{t_0} = \{i\}$,
\item $\post{t_n} = \{o\}$,
\item $\markings_0\larrow{t_0}\markings_1\larrow{t_1}\dots\larrow{t_n}\markings_n$, with no enabled transitions in $\markings_n$.
\end{itemize}
\end{definition}

\subsubsection{Structuredness}

Structured WF-nets can be characterised by their ability to be decomposed into so-called \emph{single-entry-single-exit (SESE)} regions~\cite{polyvyanyy2012structuring}, which we will refer to as \emph{process components} throughout the remainder of this paper.
A process component is a subset of nodes of a WF-net such that the subgraph induced by these nodes has a single entry node and a single exit node. 
These two nodes are called \emph{boundary} nodes as they connect the component with the rest of the model.

\begin{definition}[Boundary nodes]
Given a process component $r$, the single entry node and single exit node of $r$ are referred to as \emph{boundary nodes}.
We use the shorthand $\triangleright r$ to denote the entry node of $r$ and $r \triangleright$ to denote the exit node of $r$.
\end{definition}

\begin{definition}[Canonical process component]
A process component is \emph{canonical}, if it does not share any of its nodes with another process component. 
That is, any two canonical process components are either disjoint or one is contained in the other. 
\end{definition}

We adapt the definitions provided in~\cite{polyvyanyy2012structuring}, to further define the types of process components specific to WF-nets:


\begin{definition}[Trivial, Polygon, Bond, Rigid]
Let $C$ be a process component of a WF-net $\petri$, with boundary nodes $\triangleright C$, $C\triangleright$ $\in \places \cup \transitions$. 
\begin{itemize}
  \item $C$ is a \emph{trivial} component, iff $C$ is singleton, i.e., $C$ contains a single node such that $\triangleright C = C\triangleright$.
  \item $C$ is a \emph{polygon} component, iff there exists a sequence $(r_0,\ldots , r_n)$, $n \in \mathbb{N}$, of canonical components of $\petri$, 
  such that: 
  \begin{itemize}
    \item $C = \triangleright C$ $\cup$ $C\triangleright$ $\bigcup^{i=n}_{i=0} r_i$
    \item $(\triangleright C, \triangleright r_0), (r_n \triangleright, C\triangleright) \in \arcs$
    \item $(r_j\triangleright, \triangleright r_{j+1}) \in \arcs, 0 \leq j < n$
  \end{itemize} 
  \item C is a \emph{bond} component, iff there exists a set $R$ of canonical components of $\petri$, such that:
  \begin{itemize}
    \item $C = \triangleright C$ $\cup$ $C\triangleright$ $\bigcup_{r \in R} r$
    \item $\forall r \in R: (\triangleright C, \triangleright r), (r \triangleright, C\triangleright) \in \arcs$ 
	\item $x, y \in \places \vee x, y \in \transitions$
  \end{itemize}
  \item $C$ is a rigid component, iff $C$ is neither a trivial, nor a polygon, nor a bond component.
\end{itemize}
\end{definition}

\begin{definition}[Well-structured WF-net]
A WF-net $\petri$ is \emph{(well-)structured}, iff its set of all its canonical process components contains no rigid process component; otherwise the process model is \emph{unstructured}.
\end{definition}

\subsubsection{States and Traces}

Transitions in a WF-net can be annotated with the set of propositions that hold or not after its execution. As such, each transition can be annotated, resulting in a so-called \emph{annotated WF-net}.

\begin{definition}[Annotation]\label{def:annotation}
Given a WF-net $\petri = (\places,\transitions,\arcs,\labels,\markings_0)$, the annotation function $\annotations: \transitions \rightarrow \outcomes^2$ where $\outcomes$ is the propositional literals language, representing the set of propositions holding true after the execution of a transition.
\end{definition}

\begin{definition}[Annotated WF-net]\label{def:annotated}
Given a WF-net $\petri = (\places,\transitions,\arcs,\labels,\markings_0)$ and an annotation function $\annotations$. We refer to an annotated WF-Net as $\petri = (\places,\transitions,\arcs,\labels,\markings_0, \annotations)$.
\end{definition}

An annotation function describes what holds true after a transition is executed. After the execution of a sequence of transitions, we can compose the collection of annotations that hold true up to that point by performing a \emph{state update} after execution of each transition, as defined formally in Definition~\ref{def:litupd} below.

\begin{definition}[State update]\label{def:litupd}
Given two consistent sets of literals $\ltrls_1$ and $\ltrls_2$, representing the process state and the annotation of a transition being executed, the
  update of $\ltrls_1$ with $\ltrls_2$, denoted by $\ltrls_1\update \ltrls_2$ is a set of literals defined as
  follows:
$$ \ltrls_1\update \ltrls_2 = \ltrls_1\setminus\set{\neg{\ltrl}\mid
  \ltrl\in \ltrls_2}\cup \ltrls_2
$$
\end{definition}

We use the term \emph{trace} to refer to an execution of a WF-Net along with the associated set of literals that hold after the execution of each transition.

\begin{definition}[Trace]\label{def:trace}
Given an annotated WF-Net $\petri = (\places,\transitions,\arcs,\labels,\markings_0)$ and an execution $\exe = \seq{t_1, \ldots, t_n}$ of $\petri$, the corresponding trace is $\trace = \seq{\ltrls_1, \ldots, \ltrls_n}$, where each $\ltrls_i$ is a set of propositions and is computed as follows:

\begin{enumerate}
\item $\ltrls_0 = \emptyset$
\item $\ltrls_{i+1} = \ltrls_i\update\kwd{ann}(t_{i+1})$, for $1\le i < n$.
\end{enumerate}

Each $\ltrls_i$ in a trace represents the execution's state of the annotated WF-Net holding after executing the transition $t_i$ in the corresponding execution.
\end{definition}

In Example~\ref{ex:apm}, we illustrate an annotated WF-Net along with its possible executions and traces. In the remainder of this paper, we refer to annotated WF-nets simply as process models.


\begin{example}[Annotated Process Model]\label{ex:apm} Fig. \ref{f:p03} shows a structured process containing four tasks labelled $t_1, t_2, t_3$ and $ t_4$ and their annotations. 
The dashed lines mark the edges of the different process components, indicating polygon and bond components by \kwd{P} and \kwd{B}, respectively.
There are two bond components, of which the outer is an \kwd{AND} block with two concurrent paths. The top path is another bond component representing an \kwd{XOR} block, whereas the bottom path is a polygon representing a single task. The annotations indicate what has to hold after a task is executed. If $t_1$ is executed, then the literal $a$ has to hold in that state of the process. Table \ref{tab:01} shows the possible executions ($\Exe{B}$) and traces ($\Theta(B, \kwd{ann})$) of the model.

\end{example}

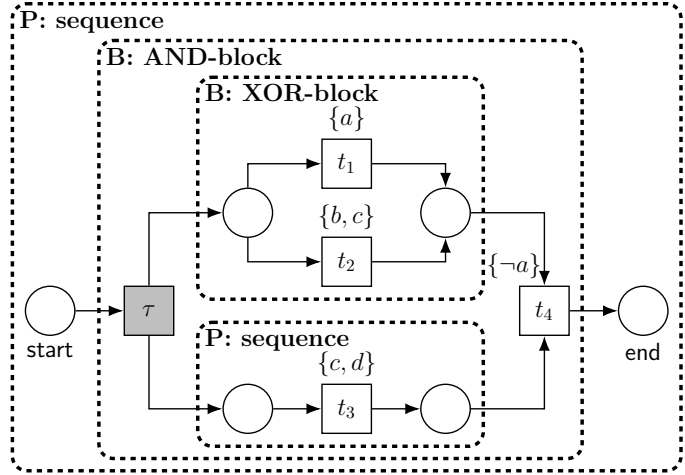
\begin{figure}[h!]
\centering
\scalebox{0.65}{\begin{tikzpicture}[thick,font=\Large]

\node[draw,align=center,circle,minimum size=1cm] (p1) at (1.0,-4.5) {};
\node[align=center,anchor=north] at (p1.south) {\kwd{start}};
\node[draw,align=center,circle,minimum size=1cm] (p2) at (5.0,-2.5) {};
\node[draw,align=center,circle,minimum size=1cm] (p3) at (5.0,-6.5) {};
\node[draw,align=center,circle,minimum size=1cm] (p4) at (9.0,-6.5) {};
\node[draw,align=center,circle,minimum size=1cm] (p5) at (9.0,-2.5) {};
\node[draw,align=center,circle,minimum size=1cm] (p6) at (13.0,-4.5) {};
\node[align=center,anchor=north] at (p6.south) {\kwd{end}};

\node[draw,fill=gray!50,align=center,minimum size=1cm] (t0) at (3.0,-4.5) {};
\node[align=center] at (t0.center) {$\tau$}; 
\node[draw,align=center,minimum size=1cm] (t1) at (7.0,-1.5) {};
\node[align=center] at (t1.center) {$t_{1}$};
\node[draw,align=center,minimum size=1cm] (t2) at (7.0,-3.5) {};
\node[align=center] at (t2.center) {$t_{2}$};
\node[draw,align=center,minimum size=1cm] (t3) at (7.0,-6.5) {};
\node[align=center] at (t3.center) {$t_{3}$};
\node[draw,align=center,minimum size=1cm] (t4) at (11.0,-4.5) {};
\node[align=center] at (t4.center) {$t_{4}$};

\draw[-{Latex[length=3mm]}] (p4) -- (11.0,-6.5) -- (t4);
\draw[-{Latex[length=3mm]}] (p1) -- (t0);
\draw[-{Latex[length=3mm]}] (p5) -- (11.0,-2.5) -- (t4);
\draw[-{Latex[length=3mm]}] (t0) -- (3.0,-2.5) -- (p2);
\draw[-{Latex[length=3mm]}] (t0) -- (3.0,-6.5) -- (p3);
\draw[-{Latex[length=3mm]}] (t4) -- (p6);
\draw[-{Latex[length=3mm]}] (p2) -- (5.0,-1.5) -- (t1);
\draw[-{Latex[length=3mm]}] (p2) -- (5.0,-3.5) -- (t2);
\draw[-{Latex[length=3mm]}] (p3) -- (t3);
\draw[-{Latex[length=3mm]}] (t3) -- (p4);
\draw[-{Latex[length=3mm]}] (t2) -- (9.0,-3.5) -- (p5);
\draw[-{Latex[length=3mm]}] (t1) -- (9.0,-1.5) -- (p5);

\node[anchor=south] at (t1.north) {\{$a$\}};
\node[anchor=south] at (t2.north) {\{$b,c$\}};
\node[anchor=south] at (t3.north) {\{$c,d$\}};
\node[xshift=-1mm,anchor=south] at (t4.north west) {\{$\neg a$\}};

\draw[dashed,line width=2pt,rounded corners=3mm] ([xshift=-0.25cm,yshift=2.75cm]p1.west |- t1.north) -- 
                                                 ([xshift=0.25cm,yshift=2.75cm]p6.east |- t1.north) -- 
                                                 ([xshift=0.25cm,yshift=-0.75cm]p6.east |- t3.south) -- 
                                                 ([xshift=-0.25cm,yshift=-0.75cm]p1.west |- t3.south) -- 
                                                 cycle;
\node[anchor=north west] at ([xshift=-0.25cm,yshift=2.75cm]p1.west |- t1.north) {\bf P: sequence};

\draw[dashed,line width=2pt,rounded corners=3mm] ([xshift=-0.5cm,yshift=2cm]t0.west |- t1.north) -- 
                                                 ([xshift=0.25cm,yshift=2cm]t4.east |- t1.north) -- 
                                                 ([xshift=0.25cm,yshift=-0.5cm]t4.east |- t3.south) -- 
                                                 ([xshift=-0.5cm,yshift=-0.5cm]t0.west |- t3.south) -- 
                                                 cycle;
\node[anchor=north west] at ([xshift=-0.5cm,yshift=2cm]t0.west |- t1.north) {\bf B: AND-block};

\draw[dashed,line width=2pt,rounded corners=3mm] ([xshift=-0.5cm,yshift=1.25cm]p2.west |- t1.north) -- 
                                                 ([xshift=0.25cm,yshift=1.25cm]p5.east |- t1.north) -- 
                                                 ([xshift=0.25cm,yshift=-0.25cm]p5.east |- t2.south) -- 
                                                 ([xshift=-0.5cm,yshift=-0.25cm]p2.west |- t2.south) -- 
                                                 cycle;
\node[anchor=north west] at ([xshift=-0.5cm,yshift=1.25cm]p2.west |- t1.north) {\bf B: XOR-block};

\draw[dashed,line width=2pt,rounded corners=3mm] ([xshift=-0.5cm,yshift=1.25cm]p3.west |- t3.north) -- 
                                                 ([xshift=0.25cm,yshift=1.25cm]p4.east |- t3.north) -- 
                                                 ([xshift=0.25cm,yshift=-0.25cm]p4.east |- t3.south) -- 
                                                 ([xshift=-0.5cm,yshift=-0.25cm]p3.west |- t3.south) -- 
                                                 cycle;
\node[anchor=north west] at ([xshift=-0.5cm,yshift=1.25cm]p3.west |- t3.north) {\bf P: sequence};
                                                                                                                                     
\end{tikzpicture}}
\caption{An annotated process}
\label{f:p03}
\end{figure}

\begin{table*}[h!]
\center
{\scriptsize
\begin{tabular}{lcl}
\multicolumn{1}{c}{$\Exe{B}$}   &\vline& \multicolumn{1}{c}{$\Theta(B, \kwd{ann})$} \\\hline
$(\kwd{start},t_1, t_3, t_4,\kwd{end})$ &\vline & $((\kwd{start}, \emptyset), (t_1, \{a\}), (t_3, \{a,c,d\}), (t_4, \{\neg a,c,d\}), (\kwd{end}, \{\neg a,c,d\}))$\\
$(\kwd{start},t_2, t_3, t_4\kwd{end})$ &\vline &  $((\kwd{start}, \emptyset), (t_2, \{b,c\}), (t_3,\{b,c,d\}), (t_4,\{\neg a,b,c,d\}), (\kwd{end},\{\neg a,b,c,d\}))$ \\
$(\kwd{start},t_3, t_1, t_4\kwd{end})$ &\vline & $((\kwd{start}, \emptyset), (t_3, \{c,d\}), (t_1, \{a,c,d\}), (t_4,\{\neg a,c,d\}), (\kwd{end},\{\neg a,c,d\}))$ \\
$(\kwd{start},t_3, t_2, t_4\kwd{end})$ &\vline & $((\kwd{start}, \emptyset), (t_3, \{c,d\}), (t_2, \{b,c,d\}), (t_4.\{\neg a,b,c,d\}), (\kwd{end},\{\neg a,b,c,d\}))$
\end{tabular}
}
\caption{Executions and traces of the annotated process in Fig. \ref{f:p03}.}\label{tab:01}
\end{table*}

\subsection{Regulatory Framework}\label{s:regulatory}
The regulatory framework describes the requirements that a process model must follow to be considered compliant. Depending on whether all, some, or none of the the executions satisfy these requirements, a different compliance level is determined for the process model. 
The regulatory requirements in a regulatory framework are described using obligations. These obligations are specified using a fragment of Process Compliance Logic (PCL), introduced by Governatori and Rotolo~\cite{PCL}.

\subsubsection{Semantics}

Obligations are evaluated over traces of a process model. We refer to the \emph{sub-segments} of a trace where an obligation is evaluated as \emph{in force intervals}, and depending on the obligation's type these in force intervals are evaluated as follows: \emph{achievement} obligations require that a state within each in force interval to satisfy the obligation's requirement, and \emph{maintenance} obligations require that each state in each interval satisfies the requirement.

\begin{definition}[Obligations]\label{def:obligations}
An obligation is a tuple $\Obl^{o}\langle \oblr, \oblt, \obld\rangle$, where $o$ is its type, and $\oblr, \oblt, \obld$ are propositional formulae representing respectively the obligation's requirement, trigger and deadline.

Given a trace $\trace = \seq{\ltrls_1, \ldots, \ltrls_n}$, $\Obl^{o}\langle \oblr, \oblt, \obld\rangle$ is either satisfied or violated as follows, depending on its type:

\begin{description}
\item[achievement]:
\begin{description}
\item[satisfied] $\forall \ltrls_i \in \trace | \ltrls_i \models \oblt$:
\begin{itemize}
\item $\exists \ltrls_j \in \trace | \ltrls_j \models \oblr$ and $i \leq j$, and
\item $\not \exists \ltrls_k \in \trace | \ltrls_k \models \obld$ and $k < j$.
\end{itemize}
\item[violated] otherwise.
\end{description}
\item[maintenance]:
\begin{description}
\item[satisfied] $\forall \ltrls_i \in \trace | \ltrls_i \models \oblt$:
\begin{itemize}
\item $\exists \ltrls_j \in \trace | \ltrls_j \models \obld$ and $i \leq j$, and
\item $\forall k | i \leq k \leq j: \ltrls_k \models \oblr$.
\end{itemize}
\item[violated] otherwise.
\end{description}
\end{description}

\end{definition}

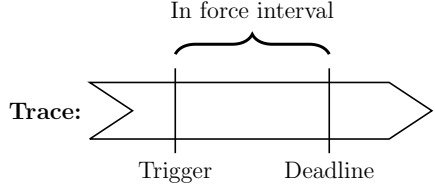
\begin{figure}[h!]
\centering
\scalebox{0.75}{\begin{tikzpicture}[thick, font=\large]

\linespread{0.75}
\newcommand{\minsize}{1.5cm}
\newcommand{\wi}{6cm}
\newcommand{\hi}{1cm}
\newcommand{\incline}{0.75cm}
\newcommand{\trigger}{0.25*\wi}
\newcommand{\deadline}{0.7*\wi}

\draw (0,0) -- ([xshift=-\incline]\wi,0) -- (\wi, 0.5*\hi) -- ([xshift=-\incline]\wi,\hi) -- (0,\hi) -- (\incline,0.5*\hi) -- (0,0);
\node[anchor=east] at (0,0.5*\hi) {\bf Trace:};

\draw ([yshift=0.25cm]\trigger,\hi) -- (\trigger,-0.25);
\node[anchor=north] at (\trigger,-0.25) {Trigger};

\draw ([yshift=0.25cm]\deadline,\hi) -- (\deadline,-0.25);
\node[anchor=north] at (\deadline,-0.25) {Deadline};

\draw [decorate,decoration={brace,amplitude=10pt},line width=1.5pt] ([yshift=0.5cm]\trigger,\hi) -- ([yshift=0.5cm]\deadline,\hi);

\node[yshift=1cm,anchor=south] at (0.5*\deadline+0.5*\trigger,\hi) {In force interval};

\end{tikzpicture}}
\caption{Conditional Obligation}\label{f:obl}
\end{figure}

Figure~\ref{f:obl} illustrates an in force interval of a local obligation within a trace. This simplistic representation shows that an in force interval is identified in a trace between two points: the \emph{trigger} and the \emph{deadline}. 

From Definition~\ref{def:obligations} it follows that states satisfying the deadline of an obligation are critical to determine whether the obligation is satisfied or violated. To avoid a scenario where an obligation being evaluated reaches the end of a trace without encountering a state satisfying the deadline (thus potentially leaving the obligation in an unknown state), we adopt Assumption~\ref{ass:final_deadline}. 

\begin{assumption}[Final Deadline]\label{ass:final_deadline}
Given a trace $\trace = \seq{\ltrls_1, \ldots, \ltrls_n}$ and an obligaton $\Obl^{o}\langle \oblr,$ $\oblt,$ $\obld\rangle$, we assume that $\ltrls_n \models \obld$ independently from the contents of $\ltrls_n$.
\end{assumption}

\subsubsection{Discussion on Multiple in Force Intervals}

A trace may contain multiple states satisfying the trigger and the deadline of an obligation, such as for instance shown in Figure~\ref{f:minfo}. 
In this example, two states satisfy the trigger, labelled \emph{trigger 1} and \emph{trigger 2}, and two states satisfying the deadline, labelled \emph{deadline 1} and \emph{deadline 2}. 
From Definition~\ref{def:obligations} it follows that each of such intervals needs to be satisfied for the trace to satisfy the obligation.

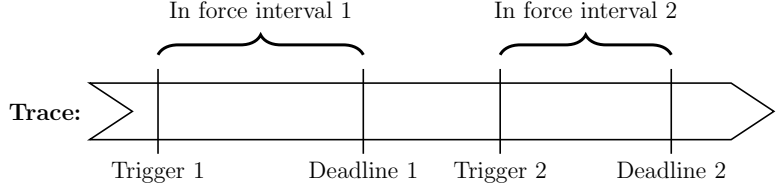
\begin{figure}
\centering
\scalebox{0.75}{\begin{tikzpicture}[thick, font=\large]

\linespread{0.75}
\newcommand{\minsize}{1.5cm}
\newcommand{\wi}{12cm}
\newcommand{\hi}{1cm}
\newcommand{\incline}{0.75cm}
\newcommand{\tOne}{0.1*\wi}
\newcommand{\dOne}{0.4*\wi}

\newcommand{\tTwo}{0.6*\wi}
\newcommand{\dTwo}{0.85*\wi}

\draw (0,0) -- ([xshift=-\incline]\wi,0) -- (\wi, 0.5*\hi) -- ([xshift=-\incline]\wi,\hi) -- (0,\hi) -- (\incline,0.5*\hi) -- cycle;
\node[anchor=east] at (0,0.5*\hi) {\bf Trace:};

\draw ([yshift=0.25cm]\tOne,\hi) -- (\tOne,-0.25);
\node[anchor=north] at (\tOne,-0.25) {Trigger 1};

\draw ([yshift=0.25cm]\dOne,\hi) -- (\dOne,-0.25);
\node[anchor=north] at (\dOne,-0.25) {Deadline 1};

\draw ([yshift=0.25cm]\tTwo,\hi) -- (\tTwo,-0.25);
\node[anchor=north] at (\tTwo,-0.25) {Trigger 2};

\draw ([yshift=0.25cm]\dTwo,\hi) -- (\dTwo,-0.25);
\node[anchor=north] at (\dTwo,-0.25) {Deadline 2};

\draw [decorate,decoration={brace,amplitude=10pt},line width=1.5pt] ([yshift=0.5cm]\tOne,\hi) -- ([yshift=0.5cm]\dOne,\hi);
\node[yshift=1cm,anchor=south] at (0.5*\dOne+0.5*\tOne,\hi) {In force interval 1};

\draw [decorate,decoration={brace,amplitude=10pt},line width=1.5pt] ([yshift=0.5cm]\tTwo,\hi) -- ([yshift=0.5cm]\dTwo,\hi);
\node[yshift=1cm,anchor=south] at (0.5*\dTwo+0.5*\tTwo,\hi) {In force interval 2};

\end{tikzpicture}}
\caption{Multiple In Force Instances}\label{f:minfo}
\end{figure}

When two in force instances intersect,the only possible intersection is shown in Figure~\ref{fig:st}. When this is the case, Lemma~\ref{lem:overlapping} states when two intervals intersect then both can be evaluated by evaluating a single interval.

\begin{figure}[h!]
\centering
\scalebox{0.75}{\begin{tikzpicture}[thick, font=\large]

\linespread{0.75}
\newcommand{\minsize}{1.5cm}
\newcommand{\wi}{8cm}
\newcommand{\hi}{1cm}
\newcommand{\incline}{0.75cm}
\newcommand{\tOne}{0.2*\wi}
\newcommand{\dOne}{0.8*\wi}

\newcommand{\tTwo}{0.4*\wi}
\draw (0,0) -- ([xshift=-\incline]\wi,0) -- (\wi, 0.5*\hi) -- ([xshift=-\incline]\wi,\hi) -- (0,\hi) -- (\incline,0.5*\hi) -- cycle;
\node[anchor=east] at (0,0.5*\hi) {\bf Trace:};

\draw ([yshift=0.25cm]\tOne,\hi) -- (\tOne,-0.25);
\node[anchor=north] at (\tOne,-0.25) {$t_1$};

\draw[line width=2pt] ([yshift=0.25cm]\dOne,\hi) -- (\dOne,-0.25);
\node[anchor=north west,align=center] at (\dOne,-0.25) {$d$};

\draw ([yshift=0.25cm]\tTwo,\hi) -- (\tTwo,-0.25);
\node[anchor=south west] at (\tTwo,0) {$t_2$};


\draw [decorate,decoration={brace,amplitude=10pt},line width=1.5pt] ([yshift=0.5cm]\tOne,\hi) -- ([yshift=0.5cm]\dOne,\hi);
\node[yshift=1cm,anchor=south] at (0.5*\dOne+0.5*\tOne,\hi) {$i_1$}; 

\draw [decorate,decoration={brace,amplitude=10pt},line width=1.5pt] ([yshift=0.5cm]\dOne,-\hi) --  ([yshift=0.5cm]\tTwo,-\hi);
\node[yshift=0cm,anchor=north] at (0.5*\dOne+0.5*\tTwo,-\hi) {$i_2$}; 

\end{tikzpicture}}
\caption{Successful Termination of 2 Overlapping In Force Instances}\label{fig:st}
\end{figure}
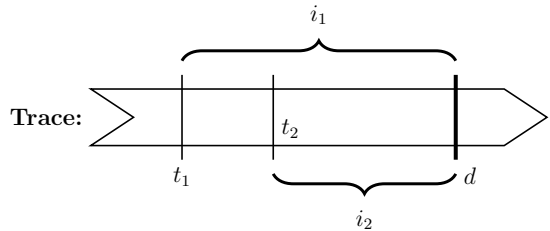

\begin{lemma}[Overlapping In Force Intervals]\label{lem:overlapping}
Consider a trace $\trace$ and a local obligation $\Obl^{o}\langle \oblr,$ $\oblt,$ $\obld\rangle$ having two in force intervals $i_1, i_2$ in $\trace$. If $i_1$ and $i_2$ overlap\footnote{Two in force intervals are considered to be overlapping if one of the intervals has its trigger point or deadline point in the other interval.}, where $i_2 \subset i_1$ (as shown in Figure~\ref{fig:st}), then:
\begin{description}
\item[o = achievement] evaluating $i_2$ provides the evaluation for both $i_1$ and $i_2$, and 
\item[o = maintenance] evaluating $i_1$ provides the evaluation for both $i_1$ and $i_2$. 
\end{description}



\end{lemma}

\begin{proof}

We prove the lemma by checking the two cases depending on the type of the obligation:
%

\noindent\begin{tabular}{@{}l p{11.55cm}}
{\bf Achievement}  & From Definition \ref{def:obligations} we know that if $i_2$ is satisfied, then the state satisfying $i_2$ would also satisfy $i_1$, and if $i_2$ is not satisfied, then the evaluation of $i_1$ becomes irrelevant as the trace would not satisfy the obligation producing these two in force intervals.
\\

{\bf Maintenance} & From Definition \ref{def:obligations} we know that a required property needs to be maintained within an in force interval. When such property needs to be maintained for the duration of two in force intervals ($i_1$ and $i_2$), and these intervals are overlapping, this is equivalent to maintain the property for the first interval, as the deadline of the first interval would also terminate the second interval. \\
\end{tabular}

\end{proof}

\subsection{Compliance Classifications}

A process checked for compliance can be classified as follows: \emph{fully compliant} if every trace of the process is compliant with the obligations composing the regulatory framework, \emph{partially compliant} if there exists at least one compliant trace, and \emph{not compliant} if none of the traces comply with the obligations.

\begin{definition}[Process Compliance]$\:$\label{def:pcompliance}
Given a process $(P, \kwd{ann})$ and a regulatory framework composed by a set of obligations $\Obls$, the compliance of $(P, \kwd{ann})$ with respect to $\Obls$ is determined as follows:
  \begin{itemize}
 \item \textbf{Full compliance} $\PCompliance[F]{(P, \kwd{ann})}{\Obls}$ if and only if\\ $\forall \trace \in \Theta(P, \kwd{ann}), \PCompliance{\trace}{\Obls}$.
 \item\textbf{Partial compliance} $\PCompliance[P]{(P, \kwd{ann})}{\Obls}$ if and only if\\ $\exists \trace \in \Theta(P, \kwd{ann}) | \PCompliance{\trace}{\Obls}$.
  \item\textbf{Non compliance} $\NPCompliance{(P, \kwd{ann})}{\Obls}$ if and only if\\ ${\neg \exists} \trace \in \Theta(P, \kwd{ann}) | \PCompliance{\trace}{\Obls}$.
  \end{itemize}
\end{definition}

\subsection{Problem Variants}\label{sec:variant_properties}

Given the components constituting a compliance problem introduced in the previous sections, we introduce three binary properties of the regulatory framework that allow to distinguish eight different variants of the compliance problem.

\bigskip
\noindent\begin{tabular}{@{}l p{13.15cm}}
{\bf 1/n} & A regulatory framework containing a single obligation ({\bf 1}) \textbf{or} a set of obligations ({\bf n}).\\
{\bf G/L} & A regulatory framework limited to contain \emph{global obligations}. Global obligations are a special case where their \emph{in-force} instance is set by default to the full trace duration, and we represent tham by omitting the trigger and deadline as follows: $\Obl^{o}\langle \oblr, \_, \_\rangle$, \textbf{or} \emph{conditional obligations}\footnote{Conditional obligations are sometimes referred to as local obligations as the counterpart to global obligations.} are allowed ({\bf L}).\\
{\bf -/+} & A regulatory framework containing obligations limited to be described using propositional literals ({\bf -}) \textbf{or} allowing such obligations to be described using propositional formulae ({\bf +}). \\
\end{tabular}
\bigskip


Each of the binary property allows an easy and a hard option:

\begin{itemize}
\item It is easier to check compliance against a single obligation ({\bf 1}) than multiple ones ({\bf n}).
\item A global obligation ({\bf G}) allows a single in force instance spanning for the entirety of the traces, making it the easier counterpart, as conditional obligations ({\bf L}) may allow multiple in force instances, which also need to be computed in each possible trace of the process model.
\item Obligations limiting their elements to propositions ({\bf -}) instead of formulae ({\bf +}) are easier to verify, as the satisfaction of their elements can be reduced to checking the execution of tasks having such elements in their annotations. Checking the satisfaction of the elements of an obligation allowing formulae always requires to consider the trace's history, making it harder.
\end{itemize}

Considering two variants of the problem sharing two of the three binary properties, it is fair to assume that the one being defined by the harder differentiating property, is at least as computationally complex as the one defined by the easier differentiating property. For instance, the variant \textbf{nG+} is at least as difficult as both \textbf{nG-} or \textbf{1G+}. This is illustrated graphically in Figure~\ref{fig:empty_lattice} below, where the potential increase in complexity is indicated by the direction of the arrows.

\begin{figure}[h!]
\centering
\scalebox{0.65}{\begin{tikzpicture}[thick,font=\LARGE]

\newcommand{\hspacing}{5cm};
\newcommand{\vspacing}{5cm};
\newcommand{\dhspacing}{2.25cm}
\newcommand{\dvspacing}{2.25cm};
\newcommand{\tw}{1.5cm}

\node (1gminus) at (0,0) {$1G-$};
\node[xshift=\hspacing,text width=\tw,align=center] (1lminus) at (1gminus.east) {$1L-$};

\node[yshift=\vspacing,text width=\tw,align=center] (ngminus) at (1gminus.north) {$nG-$};
\node[yshift=\vspacing,text width=\tw,align=center] (nlminus) at (1lminus.north) {$nL-$};

\node[xshift=\dhspacing,yshift=\dvspacing,text width=\tw,align=center] (1gplus) at (1gminus.north east) {$1G+$};
\node[xshift=\dhspacing,yshift=\dvspacing,text width=\tw,align=center] (1lplus) at (1lminus.north east) {$1L+$};

\node[yshift=\vspacing,text width=\tw,align=center] (ngplus) at (1gplus) {$nG+$};
\node[yshift=\vspacing,text width=\tw,align=center] (nlplus) at (1lplus) {$nL+$};

%
%

\draw[->,-{Latex[length=3mm]}] (1gminus) -- (1lminus);

\draw[->,-{Latex[length=3mm]}] (1gminus) -- (ngminus);
\draw[->,-{Latex[length=3mm]}] (1lminus) -- (nlminus);

\draw[->,-{Latex[length=3mm]}] (1gminus) -- (1gplus);
\draw[->,-{Latex[length=3mm]}] (1lminus) -- (1lplus);

\draw[->,-{Latex[length=3mm]}] (ngminus) -- (nlminus);
\draw[->,-{Latex[length=3mm]}] (ngplus) -- (nlplus);

\draw[->,-{Latex[length=3mm]}] (ngminus) -- (ngplus);
\draw[->,-{Latex[length=3mm]}] (nlminus) -- (nlplus);

\draw[->,-{Latex[length=3mm]}] (1gplus) -- (ngplus);
\draw[->,-{Latex[length=3mm]}] (1lplus) -- (nlplus);

\draw[->,-{Latex[length=3mm]}] (1gplus) -- (1lplus);

\end{tikzpicture}}
\caption{Compliance Complexity Lattice.}
\label{fig:empty_lattice}
\end{figure}
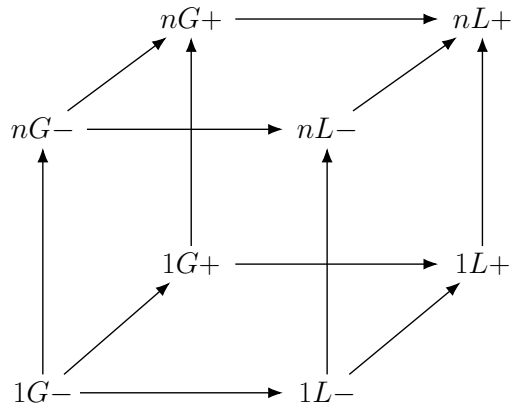

\section{Full Compliance: 1G+ is coNP-Complete}\label{sec:1g+}

The \emph{tautology problem} consists of proving whether a given propositional logic formula $\varphi$ is true in all of its possible interpretations, and its computational complexity is known to be co\textbf{NP}-complete.
In this section, we show that proving full compliance of the variant \textbf{1G+} is co\textbf{NP}-complete, by reducing the tautology problem to \textbf{1G+}.
As such, we require a WF-net representating a set of traces covering each interpretation of $\varphi$:

\begin{definition}[Interpretation WF-net construction]\label{def:ipc}
Given a propositional logic formula $\varphi$, a WF-net $\petri$, containing a set of traces whose aggregated collection of process states represent each possible interpretation of $\varphi$, is constructed as follows:


\begin{enumerate}
\item For each proposition appearing in $\varphi$ a \emph{bond} component representing an \kwd{XOR} block is created containing the following:
\begin{itemize}
\item A \emph{trivial} component with a transition having the proposition annotated.
\item A \emph{trivial} component with a transition having the \emph{negation} of the proposition annotated.
\end{itemize}
\item Construct a \emph{polygon} component representing a \kwd{SEQ} block starting with the \kwd{start} and concluded by the \kwd{end}, having between these places the following components included:
\begin{enumerate}
\item A \emph{trivial} component composed by a transition labeled \emph{init} and annotated with the set of proposition containing all the positive variants of the propositions appearing in $\varphi$, followed by: 
\item All \kwd{XOR} blocks constructed in 1. The order in which these \kwd{XOR} blocks are included is not relevant.
\end{enumerate}
\end{enumerate}

\end{definition}

Given a propositional formula $\varphi$, proving whether $\varphi$ is a \emph{tautology} can be reduced to a \textbf{1G+} problem as follows:

\begin{reduction}
\mbox{}
%
%

\bigskip
\noindent\begin{tabular}{@{}p{1.65cm} p{12.4cm}}
{\bf WF-Net} & The model is constructed as an \emph{Interpretation Process} as described in Definition~ \ref{def:ipc}. 
This construction allows to represent each possible interpretation of $\varphi$ in the traces of the model.
A generic result of the construction process detailed above is shown in Figure~\ref{fig:ipc}.
When one of the \kwd{XOR} blocks constructed in step 1 is executed, one of the transitions composing it is executed and either the positive or negative variant of a proposition is introduced in the process state. When all of the \kwd{XOR} blocks are executed, the resulting process state contains a positive or negative variant of each proposition appearing in $\varphi$, corresponding to one of the possible interpretations of the formula, hence the order in which the \kwd{XOR} are included in the \kwd{SEQ} is not relevant. The inclusion of the \emph{init} transition allows to always have in the execution's state of the WF-Net a complete interpretation for the variables of $\varphi$.\\
\end{tabular}

\noindent\begin{tabular}{@{}p{1.65cm} p{12.4cm}}
{\bf RF} & The regulatory framework is composed by the following global maintenance obligation: $\Obl^{m}\langle \varphi, \_, \_\rangle$. Note that at any step of the execution, the values of the variables composing the process state always represents a possible interpretation for the formula $\varphi$, hence by representing the obligation as a \emph{maintenance} we ask that every possible interpretation makes the formula true.\\
\end{tabular}

\end{reduction}

\begin{figure}[h!]
\centering
\scalebox{0.6}{\begin{tikzpicture}[thick,font=\large]

\node[draw,align=center,circle,minimum size=1cm] (p1) at (0,0) {};

\node[xshift=1cm,draw,align=center,minimum size=1cm,anchor=south west] (t1) at (p1.north east) {};
\node[align=center] at (t1.center) {$l_{1}$};
\node[xshift=1cm,draw,align=center,minimum size=1cm,anchor=north west] (t2) at (p1.south east) {};
\node[align=center] at (t2.center) {$\neg l_{1}$};

\node[xshift=1cm,draw,align=center,circle,minimum size=1cm,anchor=west] (p2) at (t1.east |- p1) {};

\node[xshift=1cm,draw,align=center,minimum size=1cm,anchor=south west] (t3) at (p2.north east) {};
\node[align=center] at (t3.center) {$l_{2}$};
\node[xshift=1cm,draw,align=center,minimum size=1cm,anchor=north west] (t4) at (p2.south east) {};
\node[align=center] at (t4.center) {$\neg l_{2}$};

\node[xshift=1cm,draw,align=center,circle,minimum size=1cm,anchor=west] (p3) at (t3.east |- p1) {};

\node[xshift=1cm,anchor=west] (dots) at (p3.east) {\bf \ldots};

\node[xshift=1cm,draw,align=center,circle,minimum size=1cm,anchor=west] (p4) at (dots.east) {};

\node[xshift=1cm,draw,align=center,minimum size=1cm,anchor=south west] (t5) at (p4.north east) {};
\node[align=center] at (t5.center) {$l_{n-1}$};
\node[xshift=1cm,draw,align=center,minimum size=1cm,anchor=north west] (t6) at (p4.south east) {};
\node[align=center] at (t6.center) {$\neg l_{n-1}$};

\node[xshift=1cm,draw,align=center,circle,minimum size=1cm,anchor=west] (p5) at (t5.east |- p1) {};

\node[xshift=1cm,draw,align=center,minimum size=1cm,anchor=south west] (t7) at (p5.north east) {};
\node[align=center] at (t7.center) {$l_{n}$};
\node[xshift=1cm,draw,align=center,minimum size=1cm,anchor=north west] (t8) at (p5.south east) {};
\node[align=center] at (t8.center) {$\neg l_{n}$};

\node[xshift=1cm,draw,align=center,circle,minimum size=1cm,anchor=west] (p6) at (t7.east |- p1) {};

\draw[-{Latex[length=3mm]}] (p1) -- (t1);
\draw[-{Latex[length=3mm]}] (p1) -- (t2);
\draw[-{Latex[length=3mm]}] (t1) -- (p2);
\draw[-{Latex[length=3mm]}] (t2) -- (p2);

\draw[-{Latex[length=3mm]}] (p2) -- (t3);
\draw[-{Latex[length=3mm]}] (p2) -- (t4);
\draw[-{Latex[length=3mm]}] (t3) -- (p3);
\draw[-{Latex[length=3mm]}] (t4) -- (p3);

\draw[dashed,-{Latex[length=3mm]}] (p3) -- (dots);
\draw[dashed,-{Latex[length=3mm]}] (dots) -- (p4);

\draw[-{Latex[length=3mm]}] (p4) -- (t5);
\draw[-{Latex[length=3mm]}] (p4) -- (t6);
\draw[-{Latex[length=3mm]}] (t5) -- (p5);
\draw[-{Latex[length=3mm]}] (t6) -- (p5);

\draw[-{Latex[length=3mm]}] (p5) -- (t7);
\draw[-{Latex[length=3mm]}] (p5) -- (t8);
\draw[-{Latex[length=3mm]}] (t7) -- (p6);
\draw[-{Latex[length=3mm]}] (t8) -- (p6);

\node[anchor=south] at (p1.north) {$\{l_1 \ldots l_n\}$};

\end{tikzpicture}}
  \caption{Interpretation WF-Net Construction.}\label{fig:ipc}
\end{figure}
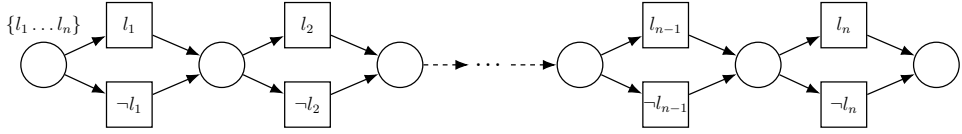

\begin{proof} 

\hspace{0pt}

\begin{enumerate}
\item From the construction of the WF-net in Definition \ref{def:ipc}, it follows that each possible trace contains in its last state a possible interpretation for the proposition composing $\varphi$.
\item From the construction of the WF-net in Definition \ref{def:ipc}, it follows that every possible interpretation of the literals included in $\varphi$ is considered in at least a state of a trace of the business process model.
\item Given the initial annotation of the task \kwd{start}, from Definition \ref{def:litupd} and the construction of the WF-net in Definition~\ref{def:ipc} it follows that the states evolving through the trace of the model represent the interpretations of the propositions composing $\varphi$ that are still possible.
\item From 1., 2., 3., and the construction of the obligation, it follows that if none of the states in every trace of the business process model falsifies the condition of the maintenance obligation constituting the regulatory framework, this means that $\varphi$ satisfies each of its possible interpretations. Therefore, $\varphi$ is a tautology.
\end{enumerate}

\vspace{0em}
\end{proof}

\section{Computational Complexity Overview}\label{sec:overview}

In this section, we show how the new complexity result integrates with existing results and how they allow to map the computational complexity of each variant of the problem for proving \emph{full}, and \emph{non} compliance.

\subsection{Existing Complexity Results for Full Compliance}


For the problem of proving full compliance, the computational complexity results already exist for five out of eight variants of the problem, as shown in Table~\ref{tab:full_complexities}. 

\begin{table}[h!]
\centering
\begin{tabular}{|c|p{8.25cm}|l|}
\hline
{\bf Variant} & {\bf Source} & {\bf Complexity Class} \\ \hline
\textbf{1G-} & Colombo Tosatto et al.~\cite{DBLP:journals/fcsc/TosattoK0KGT15} & \textbf{P} \\ \hline
\textbf{1L-} & Colombo Tosatto et al.~\cite{patternscompliance} &  \textbf{P} \\ \hline
\textbf{nL-} & Colombo Tosatto et al.~\cite{patternscompliance} &  \textbf{P} \\ \hline
\textbf{1L+} & Colombo Tosatto et al.~\cite{DBLP:journals/tsc/TosattoGK15} & co\textbf{NP}-complete \\ \hline
\textbf{nL+} & Colombo Tosatto~\cite{phd_silvano} & co\textbf{NP}-complete \\ \hline

\end{tabular}
\caption{Full Compliance Complexity}\label{tab:full_complexities}
\end{table}

\subsubsection{Complexity Overview}

The computational complexity results of each variant for proving full compliance are summarised in Figure \ref{fig:full_complexity}.

\begin{figure}[h!]
\centering
\scalebox{0.65}{\begin{tikzpicture}[thick,font=\LARGE]

\newcommand{\hspacing}{5cm};
\newcommand{\vspacing}{5cm};
\newcommand{\dhspacing}{2.25cm}
\newcommand{\dvspacing}{2.25cm};
\newcommand{\tw}{1.5cm}

\node (1gminus) at (0,0) {$1G-$};
\node[xshift=\hspacing,text width=\tw,align=center] (1lminus) at (1gminus.east) {$1L-$};

\node[yshift=\vspacing,text width=\tw,align=center] (ngminus) at (1gminus.north) {$nG-$};
\node[yshift=\vspacing,text width=\tw,align=center] (nlminus) at (1lminus.north) {$nL-$};

\node[xshift=\dhspacing,yshift=\dvspacing,text width=\tw,align=center] (1gplus) at (1gminus.north east) {$1G+$};
\node[xshift=\dhspacing,yshift=\dvspacing,text width=\tw,align=center] (1lplus) at (1lminus.north east) {$1L+$};

\node[yshift=\vspacing,text width=\tw,align=center] (ngplus) at (1gplus) {$nG+$};
\node[yshift=\vspacing,text width=\tw,align=center] (nlplus) at (1lplus) {$nL+$};

\node[anchor=east] at (1gminus.west) {\bf P};
\node[anchor=east] at (ngminus.west) {\bf P};
\node[anchor=east] at (ngplus.west) {co\textbf{NP}-c};
\node[anchor=east] at (1gplus.west) {co\textbf{NP}-c};

\node[anchor=west] at (1lminus.east) {\bf P};
\node[anchor=west] at (1lplus.east) {co\textbf{NP}-c};
\node[anchor=west] at (nlplus.east) {co\textbf{NP}-c};

\node[anchor=west] at (nlminus.east) {\bf P};

\draw[->,-{Latex[length=3mm]}] (1gminus) -- (1lminus);

\draw[->,-{Latex[length=3mm]}] (1gminus) -- (ngminus);
\draw[->,-{Latex[length=3mm]}] (1lminus) -- (nlminus);

\draw[->,-{Latex[length=3mm]}] (1gminus) -- (1gplus);
\draw[->,-{Latex[length=3mm]}] (1lminus) -- (1lplus);

\draw[->,-{Latex[length=3mm]}] (ngminus) -- (nlminus);
\draw[->,-{Latex[length=3mm]}] (ngplus) -- (nlplus);

\draw[->,-{Latex[length=3mm]}] (ngminus) -- (ngplus);
\draw[->,-{Latex[length=3mm]}] (nlminus) -- (nlplus);

\draw[->,-{Latex[length=3mm]}] (1gplus) -- (ngplus);
\draw[->,-{Latex[length=3mm]}] (1lplus) -- (nlplus);

\draw[->,-{Latex[length=3mm]}] (1gplus) -- (1lplus);

\end{tikzpicture}}
\caption{Full Compliance Complexity Lattice.}
\label{fig:full_complexity}
\end{figure}
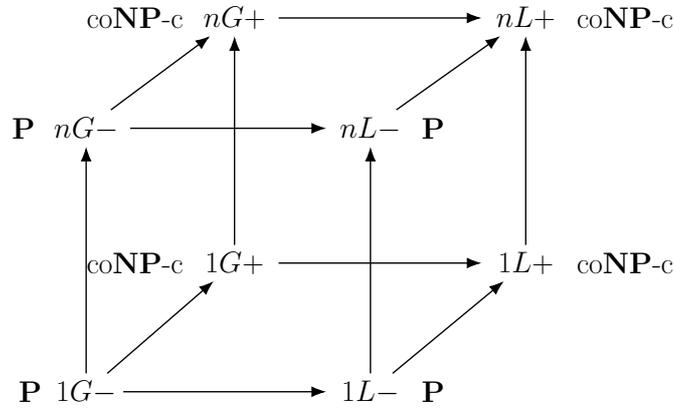

Note that Figure~\ref{fig:full_complexity} also contains results for the variants \textbf{nG-} and \textbf{nG+}, which were not explicitly provided. These results are derived using the relations between the variants of the problem shown earlier in Figure \ref{fig:empty_lattice}. 
The following relationship holds regarding the computational complexity: \textbf{1G-} $\leq$ \textbf{nG-} $\leq$ \textbf{nL-}. Therefore, by knowing that both \textbf{1G-} and \textbf{nL-} are in \textbf{P}, it follows that \textbf{nG-} is in \textbf{P} as well. Similarly, we know that with respect to computational complexity it holds that \textbf{1G+} $\leq$ \textbf{nG+} $\leq$ \textbf{nL+}. As the extremes belong to the same computational complexity class, then also \textbf{nG+} is in co\textbf{NP}-complete (shortened to co\textbf{NP}-c in Figure~\ref{fig:full_complexity}).

When proving full compliance, the computational complexity results show that variants limiting the expressivity of the elements of the obligations to atomic propositions can be solved in polynomial time, while when this limitation is absent, the complexity of solving the problem is in \textbf{NP}.

\subsection{Relations with Proving Non Compliance}

Having provided the computational complexity results for the variants of the problem of proving both full and partial compliance of process models, we can extrapolate the computational complexities of proving non compliance of the same variants. Proving non compliance of a process model is the complement of proving its partial compliance, as it requires to verify that none of the traces comply with the given regulations. Figure \ref{fig:compliance_relations} illustrates this graphically by showing that the set of partially compliant processes are complementary to the non compliant ones. Moreover the illustration also shows that fully compliant process models are also partially compliant, as they always contain at least a compliant trace.\footnote{Considering how we define process models, it cannot be the case that a model contains no traces, hence for this particular problem we can say that a \emph{for all} implies an \emph{exists}, which is generally not always the case in classical logic.}

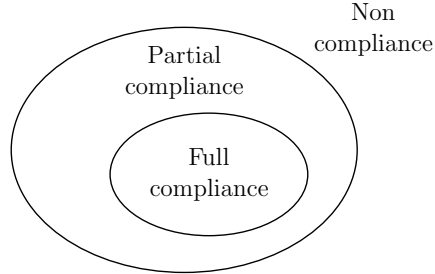
\begin{figure}[h!]
\centering
\scalebox{0.65}{\begin{tikzpicture}[thick,font=\Large]

\linespread{1}

\node[draw,ellipse,minimum width=7cm,minimum height=5cm] (partial) at (0,0) {};
\node[yshift=-0.25cm,anchor=north,align=center] at (partial.north) {Partial\\ compliance};

\node[xshift=0.5cm,yshift=-0.5cm,draw,ellipse,align=center,minimum width=4cm,minimum height=2.5cm] (full) at (0,0) {Full\\ compliance};

\node[anchor=south west,align=center] (non) at (partial.north east) {Non\\ compliance};

\end{tikzpicture}}
\caption{Compliance relations.}
\label{fig:compliance_relations}
\end{figure}
\section{Summary}\label{sec:summary}

This paper introduces a new complexity for one of the variants of the problem of proving regulatory compliance of models. When considered together with the existing computational complexity results, we can completely map the computational complexity for all these variants when the problem is to verify either full or non compliance of a model.

When the problem is to evaluate whether is partially compliant with the regulatory requirements, computational complexity results exist for 7 of the 8 variants, as reported in 
Table~\ref{tab:partial_complexities}, and the computational complexity for the variant \textbf{1L-} is still not known, as summarised in Figure~\ref{fig:partial_complexity}.

\begin{table}[h!]
\centering
\begin{tabular}{|c|p{8.25cm}|l|}
\hline
\bf Variant & \bf Source & \bf Complexity Class \\ \hline
\textbf{1G-} & Colombo Tosatto et al.~\cite{DBLP:journals/fcsc/TosattoK0KGT15} & \textbf{P} \\ \hline
\textbf{nG-} & Colombo Tosatto et al.~\cite{BPM2019_complexity} & \textbf{NP}-complete \\ \hline
\textbf{1G+} & Colombo Tosatto et al.~\cite{BPM2019_complexity} & \textbf{NP}-complete \\ \hline
\textbf{nG+} & Colombo Tosatto et al.~\cite{BPM2019_complexity} & \textbf{NP}-complete \\ \hline
\textbf{nL-} & Colombo Tosatto et al.~\cite{DBLP:journals/tsc/TosattoGK15} &  \textbf{NP}-complete \\ \hline
\textbf{1L+} & Colombo Tosatto et al.~\cite{BPM2019_complexity} & \textbf{NP}-complete \\ \hline
\textbf{nL+} & Colombo Tosatto et al.~\cite{DBLP:journals/tsc/TosattoGK15} & \textbf{NP}-complete \\ \hline
\end{tabular}
\caption{Partial Compliance Complexity}\label{tab:partial_complexities}
\end{table}

\begin{figure}[h!]
\centering
\scalebox{0.65}{\begin{tikzpicture}[thick,font=\LARGE]

\newcommand{\hspacing}{5cm};
\newcommand{\vspacing}{5cm};
\newcommand{\dhspacing}{2.25cm}
\newcommand{\dvspacing}{2.25cm};
\newcommand{\tw}{1.5cm}

\node (1gminus) at (0,0) {$1G-$};
\node[xshift=\hspacing,text width=\tw,align=center] (1lminus) at (1gminus.east) {$1L-$};

\node[yshift=\vspacing,text width=\tw,align=center] (ngminus) at (1gminus.north) {$nG-$};
\node[yshift=\vspacing,text width=\tw,align=center] (nlminus) at (1lminus.north) {$nL-$};

\node[xshift=\dhspacing,yshift=\dvspacing,text width=\tw,align=center] (1gplus) at (1gminus.north east) {$1G+$};
\node[xshift=\dhspacing,yshift=\dvspacing,text width=\tw,align=center] (1lplus) at (1lminus.north east) {$1L+$};

\node[yshift=\vspacing,text width=\tw,align=center] (ngplus) at (1gplus) {$nG+$};
\node[yshift=\vspacing,text width=\tw,align=center] (nlplus) at (1lplus) {$nL+$};

\node[anchor=east] at (1gminus.west) {\bf P};
\node[anchor=east] at (ngminus.west) {\textbf{NP}-c};
\node[anchor=east] at (ngplus.west) {\textbf{NP}-c};
\node[anchor=east] at (1gplus.west) {\textbf{NP}-c};

\node[anchor=west] at (1lminus.east) {\bf ?};
\node[anchor=west] at (1lplus.east) {\textbf{NP}-c};
\node[anchor=west] at (nlplus.east) {\textbf{NP}-c};

\node[anchor=west] at (nlminus.east) {\textbf{NP}-c};

\draw[->,-{Latex[length=3mm]}] (1gminus) -- (1lminus);

\draw[->,-{Latex[length=3mm]}] (1gminus) -- (ngminus);
\draw[->,-{Latex[length=3mm]}] (1lminus) -- (nlminus);

\draw[->,-{Latex[length=3mm]}] (1gminus) -- (1gplus);
\draw[->,-{Latex[length=3mm]}] (1lminus) -- (1lplus);

\draw[->,-{Latex[length=3mm]}] (ngminus) -- (nlminus);
\draw[->,-{Latex[length=3mm]}] (ngplus) -- (nlplus);

\draw[->,-{Latex[length=3mm]}] (ngminus) -- (ngplus);
\draw[->,-{Latex[length=3mm]}] (nlminus) -- (nlplus);

\draw[->,-{Latex[length=3mm]}] (1gplus) -- (ngplus);
\draw[->,-{Latex[length=3mm]}] (1lplus) -- (nlplus);

\draw[->,-{Latex[length=3mm]}] (1gplus) -- (1lplus);

\end{tikzpicture}}
\caption{Partial Compliance Complexity Lattice.}
\label{fig:partial_complexity}
\end{figure}
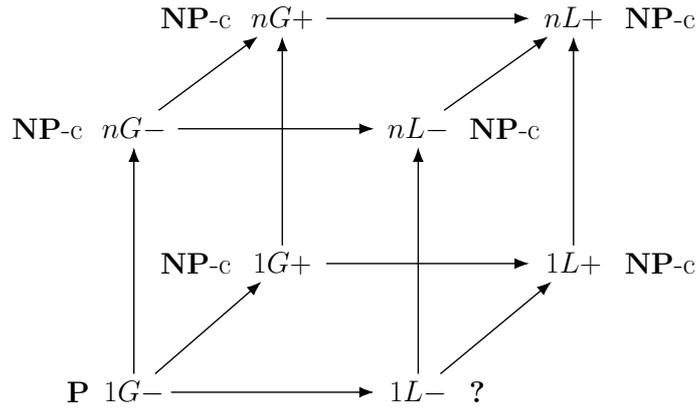


\vskip 0.2in
\bibliography{complexity}
\bibliographystyle{splncs04}

\end{document}